\documentclass[runningheads]{llncs}
\usepackage{a4wide}
\usepackage[utf8]{inputenc}
\usepackage{graphicx}
\usepackage{amsmath}
\usepackage{amssymb}
\usepackage{nicefrac}
\usepackage[disable]{todonotes}
\usepackage{enumerate}
\usepackage{subcaption}
\usepackage{caption}
\usepackage{url}
\usepackage{import}

\usepackage{adjustbox}
\usepackage{enumitem}
\usepackage{tikz}
\usepackage{siunitx}
\usepackage{multirow}
\usepackage{tabularx}
\usepackage{blindtext}
\usepackage{float}
\usepackage[breaklinks]{hyperref}
\usepackage{placeins}
\usepackage{cleveref}

\usetikzlibrary{shapes,arrows}
\usetikzlibrary{positioning}
\usetikzlibrary{calc}
\usetikzlibrary{fit}
\usetikzlibrary{shapes,decorations,arrows,calc,arrows.meta,fit,positioning}

\newcommand{\metric}[2]{\inf\limits_{#1'\in\partial_{#1},#2'\in\partial_{#2}}\|#1'-#2'\|}

\definecolor{redlight}{RGB}{235, 106, 132}
\definecolor{greenlight}{RGB}{195, 252, 186}
\definecolor{bluelight}{RGB}{178, 180, 249}

\newcounter{todocounterak}

\newcounter{todocounteras}

\newcommand{\asline}[1]
{\stepcounter{todocounteras}
 \todo[color=blue!30,inline]{AS-\thetodocounteras: #1}
}
\newcounter{todocountermk}

\newcounter{todocountermw}

\definecolor{goodbluebar}{RGB}{114, 147, 203}
\definecolor{goodorangebar}{RGB}{225, 151, 76}
\definecolor{goodgreenbar}{RGB}{132, 186, 91}
\definecolor{goodredbar}{RGB}{211, 94, 96}
\definecolor{goodblackbar}{RGB}{128, 133, 133}
\definecolor{goodpurplebar}{RGB}{144, 103, 167}
\definecolor{goodwinebar}{RGB}{171, 104, 87}
\definecolor{goodgoldbar}{RGB}{204, 194, 16}

\definecolor{goodblue}{RGB}{57, 106, 177}
\definecolor{goodorange}{RGB}{218, 124, 48}
\definecolor{goodgreen}{RGB}{62, 150, 81}
\definecolor{goodred}{RGB}{204, 37, 41}
\definecolor{goodblack}{RGB}{83, 81, 84}
\definecolor{goodpurple}{RGB}{107, 76, 154}
\definecolor{goodwine}{RGB}{146, 36, 40}
\definecolor{goodgold}{RGB}{148, 139, 61}

\begin{document}

\title{Pairwise Learning to Rank by Neural Networks Revisited:
Reconstruction, Theoretical Analysis and Practical Performance}

\author{
Marius Köppel \inst{1}\thanks{These authors contributed equally.} 
\and
Alexander Segner\inst{1}\inst{\star} 
\and
Martin Wagener\inst{1}\inst{\star} 
\and
Lukas Pensel\inst{1} 
\and
Andreas Karwath\inst{2} 
\and
Stefan Kramer\inst{1} 
}
\authorrunning{
M. Köppel et al.
}
\titlerunning{Pairwise Learning to Rank by Neural Networks Revisited}
\toctitle{Pairwise Learning to Rank by Neural Networks Revisited:
Reconstruction, Theoretical Analysis and Practical Performance}
\tocauthor{Marius~Köppel, Alexander~Segner, Martin~Wagener, Lukas~Pensel, Andreas~Karwath, Stefan~Kramer}
%
\institute{Johannes Gutenberg-Universität Mainz \\ Saarstraße 21, 55122 Mainz, Germany \\ \email{makoeppe@students.uni-mainz.de}
\and
University of Birmingham,  \\ Haworth Building (Y2), 
B15 2TT, United Kingdom \\ \email{a.karwath@bham.ac.uk}
}
\maketitle              

\begin{abstract}
We present a pairwise learning to rank approach based on a neural net, called DirectRanker, that generalizes the RankNet architecture. We show mathematically that our model is reflexive, antisymmetric, and transitive allowing for simplified training and improved performance. Experimental results on the LETOR MSLR-WEB10K, MQ2007 and MQ2008 datasets show that our model outperforms numerous state-of-the-art methods, while being inherently simpler in structure and using a pairwise approach only.
\end{abstract}

\section{Introduction}
Information retrieval has become one of the most important applications of machine learning techniques in the last years. The vast amount of data in every day life, research and economy makes it necessary to retrieve only relevant data. One of the main problems in information retrieval is the \textit{learning to rank} problem \cite{cooper1992probabilistic,Liu:2009:LRI:1618303.1618304}. Given a query $q$ and a set of documents $d_1,\dots,d_n$ one wants to find a \textit{ranking} that gives a (partial) order of the documents according to their relevance relative to $q$. Documents can in fact be instances from arbitrary sets and do not necessarily need to correspond to queries. 
\\ \\
Web search is one of the most obvious applications, however, product recommendation or question answering can be dealt with in a similar fashion. Most common machine learning methods have been used in the past to tackle the learning to rank problem \cite{Burges:2005:LRU:1102351.1102363,freund2003efficient,herbrich2000large,jiang2009learning}. In this paper we use an artificial neural net which, in a pair of documents, finds the more relevant one. This is known as the pairwise ranking approach, which can then be used to sort lists of documents. The chosen architecture of the neural net gives rise to certain properties which significantly enhance the performance compared to other approaches. We note that the structure of our neural net is essentially the same as the one of RankNet \cite{Burges:2005:LRU:1102351.1102363}. However, we relax some constraints which are used there and use a more modern optimization algorithm. This leads to a significantly enhanced performance and puts our approach head-to-head with state-of-the-art methods. This is especially remarkable given the relatively simple structure of the model and the consequently small training and testing times. Furthermore, we use a different formulation to describe the properties of our model and find that it is inherently reflexive, antisymmetric and transitive.
In summary, the contributions of this paper are:

\begin{enumerate}[itemsep=2pt]
\item	We propose a simple and effective scheme for neural network structures for pairwise ranking, called Direct\-Ranker, which is a generalization of RankNet.
\item Theoretical analysis shows which of the components of such network structures give rise to their properties and what the requirements on the training data are to make them work.
\item Keeping the essence of RankNet and optimizing it with modern methods, experiments show that, contrary to general belief, pairwise methods can still be competitive with the more recent and much more complex listwise methods.
\item Finally, we show that the method can be used in principle for classification as well. In particular, for the discovery of substructures within a single class, if they exist.
\end{enumerate}
The paper is organized as follows: We discuss different models related to our approach in \Cref{sec:related-work}. The model itself and certain theoretical properties are discussed in \Cref{sec:theory} before describing the setup for experimental tests in \Cref{sec:experi-setup} and their results in Section~\Cref{sec:experi-results}. Finally, we conclude our findings in \Cref{sec:conclusion}.

\section{Related Work}\label{sec:related-work}
There are a few fundamentally different approaches to the learning to rank problem that have been applied in the past. They mainly differ in the underlying machine learning model and in the number of documents that are combined in the cost during training. Common examples for machine learning models used in ranking are: decision trees \cite{Friedman00greedyfunction}, support vector machines \cite{cao2006adapting}, artificial neural nets \cite{learning-to-rank-from-pairwise-approach-to-listwise-approach}, boosting \cite{adapting-boosting-for-information-retrieval-measures}, and evolutionary methods \cite{Ibrahim:2017:EES:3019612.3019696}. During training, a model must rank a list of $n$ documents which can then be used to compute a suitable cost function by comparing it to the ideally sorted list. If $n=1$ the approach is called \textit{pointwise}. A machine learning model assigns a numerical value to a single document and compares it to the desired relevance label to compute a cost function. This is analogous to a classification of each document. If $n=2$ the approach is called \textit{pairwise}. A model takes two documents and determines the more relevant one. We implement this concept in our model, the DirectRanker. If $n>2$ the approach is called \textit{listwise} and the cost is computed on a whole list of sorted documents. Examples for these different approaches are \cite{cooper1992probabilistic,fuhr1989optimum,li2008mcrank} for pointwise, \cite{Burges:2005:LRU:1102351.1102363,cao2006adapting,Friedman00greedyfunction} for pairwise and \cite{learning-to-rank-from-pairwise-approach-to-listwise-approach,Ibrahim:2017:EES:3019612.3019696,Xu:2007:ABA:1277741.1277809} for listwise models. 
\\ \\
Beside our own model the focus of this paper lies mainly on the pairwise approach RankNet \cite{Burges:2005:LRU:1102351.1102363} and the listwise approach LambdaMART \cite{adapting-boosting-for-information-retrieval-measures}. RankNet is a neural net defining a single output for a pair of documents. For training purposes, a cross entropy cost function is defined on this output. LambdaMART on the other hand is a boosted tree version of LambdaRank \cite{lambdarank} which itself is based on RankNet. Here, listwise evaluation metrics $M$ are optimized by avoiding cost functions and directly defining $\lambda$-gradients
\[
 \lambda_i=\sum_jS_{ij}\left|\Delta M\frac{\partial C_{ij}}{\partial o_{ij}}\right|
\]
where $\Delta M$ is the difference in the listwise metric when exchanging documents $i$ and $j$ in a query, $C$ is a pairwise cost function, and $o_{ij}$ is a pairwise output of the ranking model. $S_{ij}=\pm1$ depending on whether document $i$ or $j$ is more relevant.
\\ \\
The main advantages of RankNet and LambdaMART are training time and performance: While RankNet performs well on learning to rank tasks it is usually outperformed by LambdaMART considering listwise metrics which is usually the goal of learning to rank. On the other hand, since for the training of LambdaMART it is necessary to compute a contribution to $\lambda_i$ for every combination of two documents of a query for all queries of a training set, it is computationally more demanding to train this model compared to the pairwise optimization of RankNet (cf \Cref{tab:runtime}). In general, multiple publications (most prominently \cite{learning-to-rank-from-pairwise-approach-to-listwise-approach}) suggest that listwise approaches are fundamentally superior to pairwise ones. As the results of the experiments discussed in \Cref{sec:compare-ranker} show, this is not necessarily the case.
\\ \\
Important properties of rankers are their reflexivity, antisymmetry and transitivity. To implement a reasonable order on the documents these properties must be fulfilled. In \cite{rigutini2008neural} the need for antisymmetry and a simple method to achieve it in neural nets are discussed. Also \cite{Burges:2005:LRU:1102351.1102363} touches on the aspect of transitivity. However, to the best of our knowledge, a rigorous proof of these characteristics for a ranking model has not been presented so far. A theoretical analysis along those lines is presented in \Cref{sec:theory} of the paper. 

\section{DirectRanker Approach}\label{sec:theory}
\begin{figure*}[ht]
	\centering

      
  \begin{tikzpicture}[scale=2, every node/.style={scale=1.0}]
    \begin{scope}[every node/.style={fill=goodredbar,circle,thin,draw,minimum width=0.75cm}]
      \foreach \xx/\yy/\lx/\ly/\nn/\nnlabel in {
        0/0/-0.0/-0.2/i1-1/i_{1_1},
        0/-0.5/-0.2/0.2/i1-2/i_{1_2},
        0/-1.25/-0.2/-0.2/i1-3/i_{1_n},
        1.0/-0.25/-0.2/-0.2/o1-1/o_{1_1}, 
        1.0/-0.75/-0.2/-0.2/o1-2/o_{1_2},
        0/-2.0/-0.2/-0.2/i2-1/i_{2_1},
        0/-2.5/-0.2/0.2/i2-2/i_{2_2},
        0/-3.25/-0.2/-0.2/i2-3/i_{2_n},
        1.0/-2.25/-0.2/-0.2/o2-1/o_{2_1}, 
        1.0/-2.75/-0.2/-0.2/o2-2/o_{2_2},
        2.5/-1.125/-0.2/-0.2/i3-1/\ominus,  
        2.5/-1.875/-0.2/-0.2/i3-2/\ominus,  
        3.5/-1.50/-0.2/-0.2/o4-1/o_{1} }{ 
        \node [label={[shift={(0.0,-0.75)}] \footnotesize\texttt{$\nnlabel$}}] (\nn) at (\xx,\yy) {} ;
      }
    \end{scope}
    \begin{scope}[>=, every edge/.style={draw=black,thin}]
      \foreach \ee/\sx/\sy\no/\nt/\bend/\bdir in {
        $w_{i_{11}-o_{11}}$/0.15/0.1/i1-1/o1-1/0/right,
        $w_{i_{12}-o_{11}}$ /0.15/-0.2/i1-2/o1-1/0/left,
        $w_{i_{13}-o_{11}}$ /0.15/-0.2/i1-3/o1-1/0/left,
        $w_{i_{11}-o_{12}}$/0.15/0.1/i1-1/o1-2/0/right,
        $w_{i_{12}-o_{12}}$ /0.15/-0.2/i1-2/o1-2/0/left,
        $w_{i_{13}-o_{12}}$ /0.15/-0.2/i1-3/o1-2/0/left,
        $w_{i_{21}-o_{21}}$/0.15/0.1/i2-1/o2-1/0/right,
        $w_{i_{22}-o_{21}}$ /0.15/-0.2/i2-2/o2-1/0/left,
        $w_{i_{23}-o_{21}}$ /0.15/-0.2/i2-3/o2-1/0/left,
        $w_{i_{21}-o_{22}}$/0.15/0.1/i2-1/o2-2/0/right,
        $w_{i_{22}-o_{22}}$ /0.15/-0.2/i2-2/o2-2/0/left,
        $w_{i_{23}-o_{22}}$ /0.15/-0.2/i2-3/o2-2/0/left,
        $w_{o_{11}-i_{21}}$/0.15/0.1/o1-1/i3-1/0/right,
        $w_{o_{12}-o_{22}}$ /0.15/-0.2/o1-2/i3-2/0/left,
        $w_{o_{21}-i_{21}}$/0.15/0.1/o2-1/i3-1/0/right,
        $w_{o_{22}-o_{22}}$ /0.15/-0.2/o2-2/i3-2/0/left,
        $w_{o_{31}-o_{41}}$ /0.15/-0.2/i3-1/o4-1/0/left,
        $w_{o_{31}-o_{41}}$ /0.15/-0.2/i3-2/o4-1/0/left}
      {
        \path [-] (\no) edge [bend \bdir=\bend] node [xshift=\sx cm,yshift=\sy cm] {} (\nt); 
      }

    \node (nn1) at (1.0,-0.5) [draw,minimum width=0.8cm,minimum height=2.0cm, fill=goodgreenbar,rounded corners=2.5mm] {\footnotesize {$nn_1$}};
    \node (nn2) at (1.0,-2.5) [draw,minimum width=0.8cm,minimum height=2.0cm, fill=goodgreenbar,rounded corners=2.5mm] {\footnotesize {$nn_2$}};

    \path [Latex-Latex,dashed] (nn1) edge [bend right=30] node  [xshift=1.1cm, text width=2cm,align=left, font=\scriptsize] {shared \\ parameters} (nn2); 

    \node (diffnode) at (4.0,-2.5) [text width=5cm,align=left,font=\scriptsize]  {Differences of corresponding output neurons from $nn_1$ and $nn_2$};
    \path [Latex-,dashed] (i3-2) edge [bend right=30] node {} (diffnode.west);  
    
    \node (outnode) at (4.0,-0.5) [text width=3.5cm,align=left,font=\scriptsize]  {Antisymmetric and sign conserving activation, zero bias in the single output neuron};
    \path [Latex-,dashed] (o4-1) edge [bend right=30] node {} (outnode.south);

    \path (i1-2) -- node[auto=false,rotate=-90]{\ldots} (i1-3);
    \path (i2-2) -- node[auto=false,rotate=-90]{\ldots} (i2-3);
    \path (i3-1) -- node[auto=false,rotate=-90]{\ldots} (i3-2);

    \draw [thick,dashed] (1.75,0) -- (1.75,-3.5);

    \node (feature) at (1.25,-3.5) [] {Feature Part};
    \node (feature) at (2.3,-3.515) [] {Ranking Part};

    \node (doc1) [left=1cm of i1-1, rotate=90] {Document 1};
    \node (doc1) [left=1cm of i2-1, rotate=90] {Document 2};
    
    \end{scope}
  \end{tikzpicture}
	\caption{Schema of the DirectRanker. $nn_1$ and $nn_2$ can be arbitrary networks (or other function approximators) as long as they give the same output for the same inputs $i_j$. The bias of the output neuron $o_{1}$ has to be zero and the activation antisymmetric and sign conserving.
}
	\label{fig:direct_ranker}
\end{figure*}
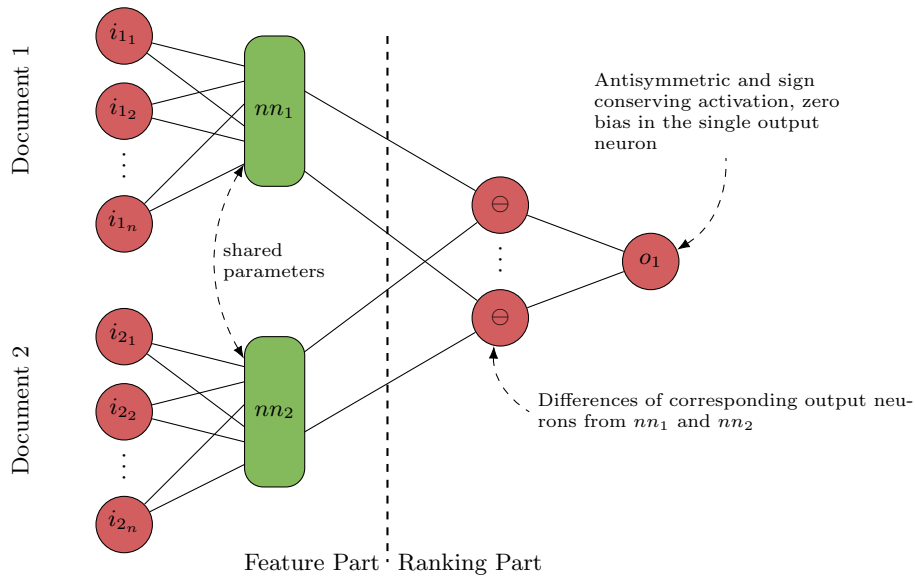
Our approach to ranking is of the pairwise kind, i.e. it takes two documents and decides which one is more relevant than the other. This approach comes with some difficulties, as, to achieve a consistent and unique ranking, the model has to define an order.
To achieve this, we implement a quasiorder $\succeq$ on the feature space $\mathcal F$ that satisfies the following conditions for all $x,y,z\in\mathcal F$:

\begin{enumerate}[label=(\Alph*),start=1,itemsep=1pt,,leftmargin=\widthof{(A)a}]
 \item Reflexivity: $x\succeq x$\label{ax:succ_antiref}
 \item Antisymmetry: $x\nsucceq y\Rightarrow y\succeq x$\label{ax:succ_antisym}
 \item Transitivity: $(x\succeq y\wedge y\succeq z)\Rightarrow x\succeq z$\label{ax:succ_trans}
\end{enumerate}
We implement such an order using a ranking function $r:\mathcal F\times\mathcal F\to\mathbb R$ by defining
\begin{align}
 x\succeq y:\Leftrightarrow r(x,y)\ge0\,.
\end{align}
The conditions (A)-(C) for $\succeq$ can be imposed in form of the following conditions for $r$:

\begin{enumerate}[label=(\Roman*),leftmargin=\widthof{(III)n}]
 \item Reflexivity: $r(x,x)=0$ \label{ax:dr_antiref} 
 \item Antisymmetry: $r(x,y)=-r(y,x)$ \label{ax:dr_antisym}
 \item Transitivity: $(r(x,y)\ge0\wedge r(y,z)\ge0)\Rightarrow r(x,z)\ge0$ \label{ax:dr_trans}
\end{enumerate}
In our case, $r$ is the output of a neural network with specific structure to fulfill the above requirements. As shown by \cite{rigutini2008neural}, the antisymmetry can easily be guaranteed in neural network approaches by removing the biases of the neurons and choosing antisymmetric activation functions. Of course, the result will only be antisymmetric, if the features fed into the network are antisymmetric functions of the documents themselves, i.e., if two documents $A$ and $B$ are to be compared by the network, the extracted features of the document pair have to be antisymmetric under exchange of $A$ and $B$. 
\\ \\
This leads to the first difficulty since it is not at all trivial to extract such features containing enough information about the documents. Our model avoids this issue by taking features extracted from each of the documents and optimizing suitable antisymmetric features as a part of the net itself during the training process. This is done by using the structure depicted in \Cref{fig:direct_ranker}.
\\ \\
The corresponding features of two documents are fed into the two subnets $nn_1$ and $nn_2$, respectively. These networks can be of arbitrary structure, yet they have to be identical, i.e. share the same structure and parameters like weights, biases, activation, etc. The difference of the subnets' outputs is fed into a third subnet, which further consists only of one output neuron with antisymmetric activation and without a bias, representing the above defined function $r$. With the following theorem we show that this network satisfies conditions \ref{ax:dr_antiref} through \ref{ax:dr_trans}:
\begin{theorem}\label{thm:transitivity}
 Let $f$ be the output of an arbitrary neural network taking as input feature vectors $x\in\mathcal F$ and returning values $f(x)\in\mathbb R^n$. Let $o_1$ be a single neuron with antisymmetric and sign conserving activation function and without bias taking $\mathbb R^n$-valued inputs. The network returning $o_1(f(x)-f(y))$ for $x,y\in\mathcal F$ then satisfies \ref{ax:dr_antiref} through \ref{ax:dr_trans}.

\end{theorem}
\begin{proof}
 Let the activation function of the output neuron be $\tau:\mathbb R\to\mathbb R$ with $\tau(-x)=-\tau(x)$ and $\text{sign}(\tau(x))=\text{sign}(x)$ as required.
\begin{enumerate}[label=(\Roman*),leftmargin=\widthof{(III)n}]
  \item If \ref{ax:dr_antisym} is fulfilled, then \ref{ax:dr_antiref} is trivially so because
  \[
   r(x,x)=-r(x,x)\forall x\in\mathcal F\Rightarrow r(x,x)\equiv0\,.
  \]
  
\item The two networks $nn_1$ and $nn_2$ are identical (as they share the same 
parameters).
Hence, they implement the same function $f:\mathcal F\to\mathbb R^n$.  The output of the complete network for the two input vectors $x,y\in\mathcal F$ is then given by: 
  \begin{align}
   r(x,y)=&\tau[w(f(x)-f(y))]=\tau[wf(x)-wf(y)]
   =:\tau[g(x)-g(y)]\,,\label{eq:g}
  \end{align}
  where $w$ is a weight vector for the output neuron and $g:\mathcal F\to\mathbb R$. This is antisymmetric for $x$ and $y$, thus satisfying the second condition \ref{ax:dr_antisym}.
  \item Let $x,y,z\in\mathcal F$, $r(x,y)\ge0,r(y,z)\ge0$, and let $g$ be defined as in \Cref{eq:g}. Since $\tau$ is required to retain the sign of the input, i.e. $\tau(x)\ge0\Leftrightarrow x\ge0$, $g(x)\ge g(y)$ and $g(y)\ge g(z)$, one finds
  \begin{align*}
   r(x,z)=&\tau[g(x)-g(z)] 
   =\tau\big[\underbrace{g(x)-g(y)}_{\ge0}+\underbrace{g(y)-g(z)}_{\ge0}\big]\ge0\,.
  \end{align*}
  Thus, $r$ is transitive and \ref{ax:dr_trans} is fulfilled.\qed
 \end{enumerate}
\end{proof}
These properties offer some advantages during the training phase of the networks for the distinction of different relevance classes:
\begin{enumerate}[label=(\roman*),leftmargin=\widthof{(iii)n}]
 \item Due to antisymmetry, it is sufficient to train the network by always feeding instances with higher relevance in one and instances with lower relevance in the other input, i.e. higher relevance always in $i_1$ and lower relevance always in $i_2$ or vice versa. 
 \item Due to transitivity, it is not necessary to
  compare very relevant and very irrelevant documents directly during 
  training. Provided that every document is
  trained at least with documents from the corresponding neighbouring
  relevance classes, the model can implicitly be trained for all
  combinations, given that all classes are represented in the training
  data.
   \item Although it might seem to be sensible to train the model such
  that it is able to predict the equality of two different documents
  of the same relevance class, the model is actually restricted when
  doing so: If the ranker is used to sort a list of documents
  according to their relevance, there is no natural order of documents
  within the same class. Hence, the result of the ranker is not
  relevant for equivalent documents. Furthermore, when only documents
  of different classes are paired in training, the optimizer employed
  in the training phase has more freedom to find an optimal solution for ranking
  relevant cases, potentially boosting the overall
  performance.
\end{enumerate}
Note that the DirectRanker is a generalization of the RankNet model \cite{Burges:2005:LRU:1102351.1102363}, which is equivalent to the DirectRanker if the activation of $o_1$ is $\tau(x)=\tanh\left(\frac x2\right)$, and if a cross entropy cost, and a gradient descent optimizer are chosen, which are free parameters in our model.
\\ \\
For simplicity, we will from now on choose the activation function to be $\tau\equiv\text{id}$. This can be done without loss of generality, since the activation function does not change the order,  if $\tau$ is sign conserving.
\\ \\
In the following, we try to put the DirectRanker on a more sound basis by analyzing some cases in which the Direct\-Ranker is able to approximate an order (given enough complexity and training samples). More precisely, we present some cases in which the following conditions are guaranteed to be met: 
\begin{enumerate}[label=(\roman*),leftmargin=\widthof{(ii)n}]
 \item There exists an order $\succeq$ satisfying \ref{ax:succ_antiref}-\ref{ax:succ_trans} on the feature space $\mathcal F$. 
 \item A given order $\succeq$ on $\mathcal F$ can be represented by the DirectRanker, i.e. there is a continuous function $r:\mathcal F\times\mathcal F\to\mathbb R$ implementing the axioms \ref{ax:dr_antiref}-\ref{ax:dr_trans} and which can be written as $r(x,y)=g(x)-g(y)$ on the whole feature space. 
\end{enumerate}
By the universal approximation theorem \cite{hornik89}, the second condition implies that $r$ can be approximated to arbitrary precision by the DirectRanker. In the following, we will discuss interesting cases, in which these assumptions are valid:
\begin{theorem}
 For every countable feature space $\mathcal F$ there exists an order $\succeq$ that is reflexive, antisymmetric, and transitive.
\end{theorem}
\begin{proof}
 By definition, for a countable set $\mathcal F$, there exists an injective function $g:\mathcal F\to\mathbb N$. Therefore, choose $x\succeq y:\Leftrightarrow g(x)\ge g(y)$ for $x,y\in\mathcal F$.\qed
\end{proof}
In fact, every sorting of the elements of countable sets satisfies \ref{ax:succ_antiref}-\ref{ax:succ_trans}, and as we show in the next theorem, it can be approximated by the direct ranker, if the set is uniformly dense:
\begin{theorem}\label{thm:isolated}
 
 Let $\succeq$ implement \ref{ax:succ_antiref}-\ref{ax:succ_trans} on an uniformly discrete feature space $\mathcal F$. Then, the DirectRanker can approximate a function that represents $\succeq$.
\end{theorem}
\begin{proof}
  First, consider $\mathcal F$ to be an infinite set. We will use the same notation as above to describe the ranking function $r$ in terms of a continuous function $g$ such that $r(x,y)=g(x)-g(y)$. Since we use neural networks to approximate $g$, referring to the universal approximation theorem \cite{hornik89}, it is sufficient to show that a continuous function $g:\mathbb R^n\to\mathbb R$ exists such that $r$ has the required properties. We will show that such a function exists by explicit construction. We can iterate through $\mathcal F$, since it is discrete and, therefore, countable. Now assign each element $x\in\mathcal F$ a value $g(x)\in\mathbb R$. Map the first value $x_0$ to 0, and then iteratively do the following with the $i$-th element of $\mathcal F$:
\begin{enumerate}[label=p1.\arabic*,start=1,leftmargin=\widthof{p1.11}]
  \item If $\exists j : j<i \wedge  x_i\succeq x_j\wedge x_j\succeq x_i$, set $g(x_i):=g(x_j)$ and continue with the next element. \label{proof3:item_one}
  \item If $\forall j$ with $j<i: x_i\succeq x_j$, set $g(x_i):=\max\limits_{j<i}g(x_j)+1$ and continue with the next element.
  \item If $\forall j$ with $j<i: x_j\succeq x_i$, set $g(x_i):=\min\limits_{j<i}g(x_j)-1$ and continue with the next element. If there are $j,k<i$ with $x_k\succeq x_i\succeq x_j$, choose an arbitrary ``largest element smaller than $x_i$'', i.e. an element $x_l\in\mathcal F,l<i$ satisfying $x_i\succeq x_l\succeq x\forall x\in\{x_j\in\mathcal F|j<i,x_j\nsucceq x_i\}$, and an arbitrary ``smallest element larger than $x_i$'', i.e. an element $x_g\in\mathcal F,g<i$ such that $x\succeq x_g\succeq x_i\forall x\in\{x_k\in\mathcal F|k<i,x_i\nsucceq x_k\}$. Then set $g(x_i):=\frac{g(x_l)+g(x_g)}2$ and continue with the next element. This is well-defined since steps 1 through 3 guarantee that every $x_l$ that can be chosen this way is mapped to the same value by $g$. Analogously for $x_g$.
 \end{enumerate}

One easily sees that this yields a function $g$ for which $g(x)\ge g(y)\Leftrightarrow x\succeq y\,\forall x,y\in\mathcal F$ and thus $r(x,y)=g(x)-g(y)\ge0\Leftrightarrow x\succeq y$.
\\ \\
Next, we expand $g$ to a continuous function in $\mathbb R^n$. Since $\mathcal F$ is uniformly discrete, $\exists\delta>0\forall i\in\mathbb N:B_\delta(x_i)\cap\mathcal F=\{x_i\}$, where $B_\delta(x_i):=\{x\in\mathbb R^n|\|x-x_i\|<\delta\}$. For every $i\in\mathbb N$ define $\tilde B_i:\overline{B_{\delta/42}(x_i)}\to\mathbb R$, $x\mapsto1-\frac{42\|x-x_i\|}\delta$. $\tilde B_i$ is obviously continuous on $\overline{B_{\delta/42}(x_i)}$. Expanding this definition to
\[
 B_i(x):=\begin{cases}
          \tilde B_i(x) & \text{if }x\in\overline{B_{\delta/42}(x_i)} \\
          0 & \text{else}
         \end{cases}
\]
allows us to define a function $g_c:\mathbb R^n\to\mathbb R$, $x\mapsto\sum_{i=1}^\infty g(x_i)B_i(x)$ which results in the same value as $g$ for all relevant points $x_i$. This can easily be checked, since $B_n(x_m)=\delta_{mn}$ (using the Kronecker-delta). Thus, it still represents $\succeq$ on $\mathcal F$. $B_i$ is continuous since $B_i|_{\overline{B_{\delta/42}}}=\tilde B_i$ and $B_i|_{\mathbb R^n\backslash B_{\delta/42}}\equiv0$ are continuous and, therefore, $B_i$ is continuous on the union of these closed sets. We now show  that $g_c$ is continuous using the $\varepsilon$-$\delta$-definition:
\\ \\
Let $\varepsilon>0$ and $x\in\mathbb R^n$. If there is no $n\in\mathbb N$ such that $x\in\overline{B_{\delta/42}(x_i)}$, we can choose $\tilde\delta>0$ such that $B_{\tilde\delta}(x)\cap\overline{B_{\delta/42}(x_i)}=\emptyset$ $\forall n\in\mathbb N$ since $\mathcal F$ is uniformly discrete. Therefore, $g_c|_{B_{\tilde\delta}}\equiv0$ and $|g_c(\tilde x)-g_c(x)|=0<\varepsilon\forall\tilde x\in B_{\tilde\delta}(x)$.
\\ \\
If there is such an $n$, then $B_{\delta/4}(x)\cap \overline{B_{\delta/42}(x_i)}$ is non-empty, if and only if $n=i$. Hence, we can choose $\frac\delta4>\tilde\delta>0$, such that $|g_c(\tilde x)-g_c(x)|<\varepsilon\forall\tilde x\in B_{\tilde\delta}(x)$ since $g_c|_{B_{\delta/4}(x_i)}=g(x_i)\cdot B_i|_{B_{\delta/4}(x_i)}$ is clearly continuous. Therefore, for every $\varepsilon>0$ and $x\in\mathbb R^n$ we can find $\tilde\delta>0$ such that $|g_c(\tilde x)-g_c(x)|<\varepsilon\forall\tilde x\in B_{\tilde\delta(x)}$, i.e., $g_c$ is continuous.
\\ \\
If $\mathcal F$ is finite with $N$ elements, set $g(x_k)=0$ for $k>N$. Then the proof works analogously as the above.
\qed
\end{proof}
Therefore, it theoretically feasible to successfully train the DirectRanker on any finite dataset, and consequently on any real-world dataset. However, the function $g$ might be arbitrarily complex depending on the explicit order. In real-world applications, the desired order is usually not discrete and the task at hand is to predict the order of elements {\em not} represented in the training data. In the following, we give a reasonably weak condition for which an order $\succeq$ can be approximated by the DirectRanker on more general feature spaces:

\begin{theorem}\label{thm:open}
 Let $\succeq$ implement \ref{ax:succ_antiref}-\ref{ax:succ_trans} and $\mathcal F\subset\mathbb R^n$ be convex and open. For every $x\in\mathcal F$ define 
 \begin{gather*}
 \mathcal P_x:=\{y\in\mathcal F|x\nsucceq y\},~\mathcal N_x:=\{y\in\mathcal F|y\nsucceq x\},~
 \partial_x:=\{y\in\mathcal F|x\succeq y\wedge y\succeq x\}\,.
 \end{gather*}
 Furthermore, let $(\mathcal F/\sim,d)$ be a metric space, where $x\sim y\Leftrightarrow y\in\partial_x$ and 
 \[
  d(\partial_x,\partial_y)=\metric{x}{y}. 
 \]
 Then, the DirectRanker can approximate $\succeq$ if $\mathcal P_x$ and $\mathcal N_x$ are open for all $x\in\mathcal F$.

\end{theorem}
\begin{proof}
 First, note that $\partial_x=\mathcal F\backslash(\mathcal P_x\cup\mathcal N_x)$ and that $\mathcal P_x\cap\mathcal N_x=\emptyset$ because of the antisymmetry of $\succeq$, dividing $\mathcal F$ into three distinct subsets $\mathcal P_x,\mathcal N_x,\partial_x$. The relation $x\sim y\Leftrightarrow y\in\partial_x$ defines an equivalence relation with equivalence classes $[x]=\partial_x$. \\
 According to the universal approximation theorem it is sufficient to show that a continuous function $g:\mathcal F\to\mathbb R$ exists such that $\forall x,y\in\mathcal F:x\succeq y\Leftrightarrow g(x)\ge g(y)$. Again this will be done by explicit construction: \\
 Now, define $g:\mathcal F\to\mathbb R$ by first choosing an arbitrary $x_0\in\mathcal F$. For all points, set
 \[
  g(x)=\begin{cases}
   d(\partial_{x_0},\partial_x) & \text{if }x\in\mathcal P_{x_0} \\ 
   -d(\partial_{x_0},\partial_x) & \text{if }x\in\mathcal N_{x_0} \\
   0 & \text{if }x\in\partial_{x_0}
  \end{cases}~~\forall x\in\mathcal F
\]
 First, let us show, that $g$ is continuous: \\
 Consider $d_{\partial_{x_0}}:(\mathcal F/\!\!\!\sim)\to\mathbb R$, $\partial_x\mapsto d(\partial_{x_0},\partial_x)$. Since $d$ as a metric is continuous, we know that $d_{\partial_{x_0}}$ is, too. Thus, for every $\partial_x\in\mathcal F/\!\!\!\sim$ and every $\varepsilon>0$ there is a $\delta>0$ such that for every $\partial_y$ with $d(\partial_x,\partial_y)<\delta$ $|d_{\partial_{x_0}}(\partial_x)-d_{\partial_{x_0}}(\partial_y)|<\varepsilon$. Therefore, for a given $\varepsilon>0$ and $x\in\mathcal F$ we can always choose $\delta$ such that for every $y\in\mathcal F$ with $\|x-y\|<\delta$ the above holds for $\partial_x$ and $\partial_y$. Then, $|\,|g(x)|-|g(y)|\,|=|d_{\partial_{x_0}}(\partial_x)-d_{\partial_{x_0}}(\partial_y)|<\varepsilon$. Thus, $|g|$ is continuous, which means that $g$ is continuous on $\mathcal P_{x_0}\cup\partial_{x_0}=\mathcal F\backslash\mathcal N_{x_0}$ and $\mathcal N_{x_0}\cup\partial_{x_0}=\mathcal F\backslash\mathcal P_{x_0}$. Since these two sets are closed in $\mathcal F$, this implies that $g$ is continuous on their union $\mathcal F=\mathcal P_{x_0}\cup\partial_{x_0}\cup\mathcal N_{x_0}$. 
\\ \\
 Finally, we need to show now that $r(x,y)=g(x)-g(y)$ represents $\succeq$, i.e. $x\succeq y\Leftrightarrow g(x)\ge g(y)$: \\
 For the case $x\succeq x_0\succeq y$ the equivalence is obvious from the definition of $g$. For the remaining cases concerning $x\succeq y\Rightarrow g(x)\ge g(y)$, suppose there are elements $x,y\in\mathcal P_{x_0}$ with $x\succeq y$ and $g(x)<g(y)$. We can then choose a continuous curve $\gamma:[0,1]\to\mathcal F$ with $\gamma(0)=x_0\in\mathcal N_y$ and $\gamma(1)=x\in\mathcal P_y\cup\partial_y$. Choose $t_0:=\sup\{t\in[0,1]|\gamma(t)\in\mathcal N_y\}$. Then, for every neighborhood $U$ of $t_0$, there are $t',t''\in U$, such that $\gamma(t')\in\mathcal N_y$, $\gamma(t'')\notin\mathcal N_y$. Therefore, $\gamma(t_0)$ is a boundary point of $N_y$ and therefore no element of $N_y$. Also, since $\mathcal P_y\cap\mathcal N_y=\emptyset$, $\gamma(t_0)$ is no inner point of $P_y$ and therefore no element of $P_y$. Thus $\gamma(t_0)=:\tilde y\in\mathcal F\backslash(\mathcal P_y\cup\mathcal N_y)=\partial_y$. If $y\succeq x$, $t_0$ may equal 1. \\
 Since this holds for any continuous curve $\gamma$ with the given boundary conditions, we can choose $\gamma(t)=tx+(1-t)x_0$. The length of $\gamma|_{[0,t_0]}$ is then given by
 \begin{align*}
  \|\tilde y-x_0\|=&L(\gamma|_{[0,t_0]})=\int_0^{t_0}\|\dot\gamma(t)\|dt\le\int_0^1\|\dot\gamma(t)\|dt=\|x-x_0\|
\end{align*}
The same argument holds when replacing $x$ and $x_0$ by arbitrary $x'\in\partial_x$ and $x_0'\in\partial_{x_0}$. Therefore, $\forall x'\in\partial_x,x_0'\in\partial_{x_0}\exists \tilde y_{x',x_0'}\in\partial_y:\|y_{x',x_0'}-x_0'\|\le\|x'-x_0'\|$, and therefore
 \begin{align*}
  g(y)=&d(\partial_y,\partial_{x_0})=\metric y{x_0}\le\metric x{x_0}=d(\partial_x,\partial_{x_0})=g(x)
 \end{align*}
 which contradicts the assumption that $g(x)<g(y)$. Following a similar chain of reasoning, the case $g(x)<g(y)$ for $x,y\in\mathcal N_{x_0}$ leads to contradictions. 
 \\ \\
 To see that $g(x)\ge g(y)\Rightarrow x\succeq y$, suppose $\exists x,y\in\mathcal F, g(x)\ge g(y), x\nsucceq y$. Because of the antisymmetry of $\succeq$, this implies $y\succeq x$, which (as shown above) leads to $g(y)\ge g(x)$. This however, is only possible if $g(x)=g(y)$. 
 \\ \\
 Analogous to the above, for every curve $\gamma(t)=t x'+(1-t)x_0'$ for $t\in[0,1]$ with  $x'\in\partial_x$, $x_0'\in\partial_{x_0}$, there exists a $\tilde y_{x',x_0'}\in\partial_y$ and $t_0\in]0,1[$ such that $\gamma(t_0)=\tilde y_{x',x_0'}$. We also know that
 \begin{align*}
  \|x'-x_0'\|=&\int_0^1\|\dot\gamma(t)\|dt=\int_0^{t_0}\|\dot\gamma(t)\|dt+\int_{t_0}^1\|\dot\gamma(t)\|dt=\|\tilde y_{x',x_0'}-x_0'\|+\|x'-\tilde y_{x',x_0'}\|
 \end{align*}
 However, $g(x)=g(y)\Leftrightarrow\metric x{x_0}=\metric y{x_0}$ implies
 \begin{align*}
  \metric y{x_0}\le&\inf\limits_{x'\in\partial_x,x_0'\in\partial_{x_0}}\|\tilde y_{x',x_0'}-x_0'\|\le\metric x{x_0}\le\metric y{x_0}
 \end{align*}
 This means, that $\inf\limits_{x'\in\partial_x,x_0'\in\partial_{x_0}}\|\tilde y_{x',x_0'}-x_0'\|=\metric x{x_0}$, which implies 
 \begin{align*}
  \metric x{x_0}=&\inf\limits_{x'\in\partial_x,x_0'\in\partial_{x_0}}(\|\tilde y_{x',x_0'}-x_0'\|+\|x'-\tilde y_{x',x_0'}\|) \\
  \ge&\inf\limits_{x'\in\partial_x,x_0'\in\partial_{x_0}}\|\tilde y_{x',x_0'}-x_0'\|+\inf\limits_{x'\in\partial_x,x_0'\in\partial_{x_0}}\|x'-\tilde y_{x',x_0'}\| \\
  \ge&\metric x{x_0}+\inf\limits_{x'\in\partial_x,x_0'\in\partial_{x_0}}\|x'-\tilde y_{x',x_0'}\|
 \end{align*}
 \begin{gather*}
  \Rightarrow d(\partial_x,\partial_y)=\metric xy \le\inf\limits_{x'\in\partial_x,x_0'\in\partial_{x_0}}\|x'-\tilde y_{x',x_0'}\|\le0 \\
  \Rightarrow d(\partial_x,\partial_y)=0
 \end{gather*}
 Which contradict the assumption $x\nsucceq y$ that implies $\partial_x\ne\partial_y$.
 \qed
\end{proof}
If there are no two documents with the same features but different relevances, any finite dataset can be extended to $\mathbb R^n$ such that the conditions for \Cref{thm:open} are met. In real-world applications, i.e. applications with noisy data, it is in general possible that $\mathcal P_x$, $\mathcal N_x$, and $\partial_x$ blur out and overlap. In this case, it is of course impossible to find any function that represents $\succeq$. However, the Di\-rect\-Ran\-ker still ought to be able to find a ``best fit'' of a continuous function that maximizes the predictive power on any new documents, even if some documents in the training set are mislabeled. Experiments investigating this are discussed in \Cref{sec:results_datasets}.

\section{Experimental Setup}\label{sec:experi-setup}
To evaluate the performance of our model and to compare it to other learning to rank approaches, we employ common evaluation metrics and standard datasets (MSLR-WEB10K, MQ2007, MQ2008 \cite{mslr}). Furthermore we use synthetic data to investigate the dependence of the performance on certain characteristics of the data. Reliable estimates for the performance are gained by averaging over different splits of the dataset using cross-validation on the predefined folds from the the original publications and are compared to other models. In all tests,  we use the TensorFlow library \cite{tensorflow2015-whitepaper} and its implementation of the Adam-Optimizer \cite{adam}.
\\ \\
In \Cref{sec:letor} we briefly describe the structure of the LETOR datasets and in \Cref{sec:gridsearch} how the models are evaluated. In \Cref{sec:setup_datasets} we illustrate how the  synthetic datasets are generated and analyzed. For evaluating different models we apply the commonly used metrics NDCG and MAP. Shown in \Cref{sec:ndcg} and \Cref{sec:map}.

\subsection{The NDCG Metric}\label{sec:ndcg}
In the field of learning to rank, a commonly used measure for the performance of a model is the normalized discounted cumulative gain of top-$k$ documents retrieved (NDCG@$k$). This metric is based on the discounted cumulative gain of top-$k$ documents (DCG@$k$):

\begin{equation}
\text{DCG@}k = \sum_{i=1}^{k} \frac{2^{r(d_i)} - 1}{log_2(i + 1)}
\,,\notag
\end{equation}
where $d_1, d_2, ..., d_n$ is the list of documents sorted by the model with respect to a single query and $r(d_i)$ is the relevance label of document $d_i$. The NDCG@$k$ can be computed by dividing the DCG@$k$ by the ideal (maximum) discounted cumulative gain of top-$k$ documents retrieved (IDCG@$k$), i.e. the DCG@$k$ for a perfectly sorted list of documents is defined as $\text{NDCG@}k = \frac{\text{DCG@}k}{\text{IDCG@}k}$. 

\subsection{The MAP Metric}\label{sec:map}
A frequently used alternative to the NDCG is the mean average precision (MAP). For this metric the precision for the top-$k$ documents of query $q$ is introduced. Since this only makes sense for binary classes, a multiclass system has to be binarized such that $r(d_i)=1$ indicates a relevant, and $r(d_i)=0$ indicates an irrelevant document:
\[
\text{P}_q @ k = \frac1k{\sum_{i=1}^{k} r(d_i)}
\]
Now one needs to calculate the average precision over all \textit{relevant} documents.
\[
\text{AvgP}_q = \frac1{n\cdot P@n}{\sum_{k=1}^{n} r(d_k)\text{P}_q @ k}
\]
where $n$ is the number of documents in the query $q$. Finally, the MAP is given by the mean over all queries:
\[
\text{MAP} = \frac1Q{\sum_{q=1}^{Q} \text{AvgP}_q}
\]
where $Q$ is the number of queries in a data set.

\subsection{The LETOR Datasets}\label{sec:letor}
The {\em Microsoft Learning to Rank Datasets} (LETOR) and especially the MSLR--WEB10K set are standard data sets that are most commonly used to benchmark {\em learning to rank} models. The dataset consists of 10,000 queries and is a subset of the larger MSLR--WEB30K dataset. Each instance in the dataset corresponds to a query-document pair which is characterized by 136 numerical features. Additionally, relevance labels from 0 (irrelevant) to 4 (perfectly relevant) indicate the relevance of the given document with respect to the query. Ranking documents according to their relevance is often simplified by binarizing the relevance labels using an appropriate threshold, as is done by  \cite{Ibrahim:2017:EES:3019612.3019696,Ibrahim2018}. In our case, we map relevance labels $\ge2$ to 1 and relevance labels $\le1$ to 0. We use this approach to compare our model to others. 
\\ \\
Additionally we evaluate the different algorithms on the much smaller MQ2007 and MQ2008 datasets. These are similar in structure to the MSLR--WEB10K set with some minor differences: The relevance labels range from 0 to 2 and each document consists of 46 features. In this case we binarize the relevance labels by mapping labels $\ge1$ to 1 and relevance labels $=0$ to 0.

\subsection{Grid search and LETOR Evaluation}\label{sec:gridsearch}
We perform grid searches for hyperparameter optimization of our model.
The grid searches were performed using the {\em GridSearchCV} class implemented in the {\em scitkit-learn} library \cite{scikit-learn}. One grid search was done to optimize the NDCG@$10$ and one optimizing the MAP. For each hyperparameter point a 5-fold cross validation (internal) was performed on each training set on each of the 5 predefined folds of the datasets. The models were then again trained using the best hyperparameters using the entire training set before averaging the performance on independent test sets over all 5 folds (external).
Before training the model the data was transformed in such a way that the features are following a normal distribution with standard deviation of $\nicefrac13$.
\\ \\
For benchmarking the results, the most common {\em learning to rank} algorithms were also trained and tested with the same method as described above. The implementation of these algorithms are taken from the RankLib library implemented in the Lemur Project \cite{lemur}. The algorithms are: RankNet \cite{Burges:2005:LRU:1102351.1102363}, AdaRank \cite{Xu:2007:ABA:1277741.1277809}, LambdaMART \cite{adapting-boosting-for-information-retrieval-measures}, and ListNet \cite{learning-to-rank-from-pairwise-approach-to-listwise-approach}.
\\ \\
The used datasets contain queries with only non relevant documents for which the evaluation metrics are not defined. Consequently we exclude those queries from the data. Furthermore, there are queries with less than 10 documents. For such queries with $k<10$ documents the NDCG@$k$ is evaluated during the tests.

\subsection{Synthetic Data Generation and Evaluation}\label{sec:setup_datasets}
To study how the DirectRanker performs for differently structured datasets, synthetic data with different characteristics were created and evaluated. To achieve comparability between the different sets, all datasets have the following properties in common: 
\begin{enumerate}[label=(\roman*),leftmargin=\widthof{(iii)n}]
 \item The dataset consists of separate training and test sets which are generated independently,  but with the same parameters. 
 \item For each relevance class, the features follow a Gaussian distribution in feature space with a constant, but random mean between 0 and 100, and a constant, but random standard deviation between 50 and 100. 
 \item Except for a test in which the performance depending on the size of the dataset is studied, all training sets consist of $10^5$ and all test sets consist of $10^4$ documents. 
 \item During training, $r(d_i)(1-o_1(d_i,d_j))^2$ is applied as the cost function as pairs are constructed such that $r(d_i)-r(d_j)=1$.
\end{enumerate}
For the different tests, one parameter describing the dataset is changed and evaluated for different values, keeping all other parameters fixed. These parameters include the size of the training set, the number of relevance classes, the number of features, and noise on the labels. The noise for the labels is generated by assuming a Gaussian for each label with variable standard deviation and rounding to the next integer. This allows for testing larger degrees of confusion between more distant relevance classes. 
\\ \\
The general procedure for the experiments is the following: 
\begin{enumerate}[label=(\arabic*)]
 \item A dataset with the respective parameters is created. 
 \item The DirectRanker is trained on the generated training set using our framework. 
 \item The trained ranker is tested on the generated test set, again using our framework. For this, 50-150 random samples are drawn from the test set. This subset is then sorted using the trained model and the NDCG@20 is calculated. The whole test is repeated 50 times and the average value of NDCG@20 over these 50 random subsets is calculated. 
 \item These three steps are repeated at least four more times to determine a mean value $\mu$ for the NDCG@20, averaged over different datasets with the same characteristics. The standard error is calculated as an uncertainty $\Delta\mu$ of $\mu$.
\asline{A link maybe?}
\end{enumerate}
In the plots showing our test results (\Cref{fig:class_dependence}, \Cref{fig:feature_dependence}, \Cref{fig:size_dependence}, \Cref{fig:mislabeling}), every data point is the result of applying these four steps for one choice of the dataset parameters. $nn_1$ and $nn_2$ consist of a hidden layer with 70 neurons and another one with as many neurons as there are relevance classes. The results of these tests are discussed in \Cref{sec:results_datasets}.

\section{Experimental Results}\label{sec:experi-results}
In this section we present the experimental results. First, we compare our model with the commonly used ones (\Cref{sec:compare-ranker}). Furthermore, we give an outline of the sensitivity on different dataset properties (\Cref{sec:results_datasets}) and then briefly describe a method how pairwise approaches can be used for classification problems (\Cref{sec:rank-classification}).

\subsection{Comparison to other Rankers}\label{sec:compare-ranker}
\begin{table}[htb]
 \centering
 \caption{Performance comparison for different rankers on multiple Letor datasets. The values for ES-Rank, IESR-Rank and IESVM-Rank are taken from \cite{Ibrahim2018}. These values are marked with itelic. LambdaMart was boosted using the corresponding evaluation metric during training.}
 \makebox[\textwidth][c]{
 \begin{tabularx}{\textwidth}{lXXXXXX}
  \hline
  & \multicolumn{2}{c}{MSLR-WEB10K} &  \multicolumn{2}{c}{MQ2008} &  \multicolumn{2}{c}{MQ2007} \\
  
  Algorithm & $\langle\text{NDCG}\rangle$ & $\langle\text{MAP}\rangle$ & $\langle\text{NDCG}\rangle$ & $\langle\text{MAP}\rangle$ & $\langle\text{NDCG}\rangle$ & $\langle\text{MAP}\rangle$\\
  \hline
  DirectRanker & \SI{0.440 \pm 0.004}{} & \SI{0.365 \pm 0.003}{} & \SI{0.72 \pm 0.012}{} & \SI{0.636 \pm 0.011}{} & \SI{0.540 \pm 0.01}{} & \SI{0.534 \pm 0.009}{}\\

  RankNet & \SI{0.157 \pm 0.003}{} & \SI{0.195 \pm 0.002}{} & \SI{0.716 \pm 0.011}{} & \SI{0.642 \pm 0.010}{} & \SI{0.525 \pm 0.011}{} & \SI{0.525 \pm 0.007}{}\\

  ListNet & \SI{0.157 \pm 0.003}{} & \SI{0.192 \pm 0.002}{} & \SI{0.719 \pm 0.010}{} & \SI{0.647 \pm 0.006}{} & \SI{0.526 \pm 0.010}{} & \SI{0.525 \pm 0.009}{}\\

  LambdaMart & \SI{0.476 \pm 0.003}{} & \SI{0.366 \pm 0.003}{} & \SI{0.723 \pm 0.007}{} & \SI{0.624 \pm 0.006}{} & \SI{0.531 \pm 0.012}{} & \SI{0.51 \pm 0.011}{}\\

  AdaRank & \SI{0.400 \pm 0.016}{} & \SI{0.322 \pm 0.01}{} & \SI{0.722 \pm 0.010}{} & \SI{0.653 \pm 0.009}{} & \SI{0.526 \pm 0.010}{} & \SI{0.527 \pm 0.01}{}\\
  
  \textit{ES-Rank} & \textit{0.382} & \textit{0.570} & \textit{0.507} & \textit{0.483} & \textit{0.451} & \textit{0.470}\\

  \textit{IESR-Rank} & \textit{0.415} & \textit{0.603} & \textit{0.517} & \textit{0.494} & \textit{0.455} & \textit{ 0.473}\\

  \textit{IESVM-Rank} & \textit{0.224} & \textit{0.457} & \textit{0.498} & \textit{0.473} & \textit{0.436} & \textit{0.456}\\
  \hline
 \end{tabularx}
 }

 \label{tab:rankercomparison1}
\end{table}
In \Cref{tab:rankercomparison1} the results of different models on the datasets discussed in \Cref{sec:letor} are shown. On the MQ2007 and MQ2008 datasets the differences between the models are insignificant ($0.54\sigma$ difference in the NDCG@10 between the best and worst performing algorithms on the MQ2008 dataset) making the experiments on these sets inconclusive. However, the results on the MSLR--WEB10K set differ significantly. Here LambdaMart outperforms the DirectRanker by $7.2 \sigma$ on the NDCG@10. On the MAP however, the difference is only $0.2 \sigma$. It is important to note that for LambdaMart the model was explicitly boosted on NDCG@10 and MAP respectively for the two performance values while the DirectRanker uses a cost function independent of the evaluation metric. On the MSLR--WEB10K set the DirectRanker outperforms all other methods by at least $2.4 \sigma$ (NDCG@10) or $3.2\sigma$ (MAP).

\begin{table}[htb]
 \centering
 \caption{Comparing the average run times over the five folds of the MSLR--WEB10K data set in seconds. The values with $*$ were trained on the machine mentioned in the text. The values with $\dagger$ were trained using RankLib. The values with $\ddagger$ are taken from \cite{Ibrahim2018}.}
 \begin{tabular}{lc}
  \hline
  Algorithm & time in seconds\\
  \hline
  DirectRanker$^{*}$ & \SI{151.94 \pm 0.41}{}\\

  RankNet$^{*\dagger}$ & \SI{2215 \pm 351}{}\\
  
  AdaRank$^{*\dagger}$ & \SI{1261 \pm 50}{}\\
  
  LambdaMART$^{*\dagger}$ & \SI{2664 \pm 234}{}\\
  
  ListNet$^{*\dagger}$ & \SI{3947 \pm 481}{}\\
  
  ES-Rank$^{\ddagger}$ & \SI{1800 }{}\\

  IESR-Rank$^{\ddagger}$ & \SI{1957}{} \\

  IESVM-Rank$^{\ddagger}$ & \SI{34209}{} \\
  \hline
 \end{tabular}
 
 \label{tab:runtime}
\end{table}
 To demonstrate the simplicity of the DirectRanker, we present experiments on the runtime for the model training in \Cref{tab:runtime}. All tests performed by us have been conducted on an Intel\textsuperscript\textregistered\ Core\texttrademark\ i7-6850K CPU @ 3.60GHz using the above mentioned MSLR--WEB10K dataset averaging over the five given folds. Our model was trained using TensorFlow \cite{tensorflow2015-whitepaper}, contrary to the other implementations. This makes the comparison of the run times difficult, however, we also reimplemented RankNet using TensorFlow in the same way as our model. Here it can be seen that the runtime of the DirectRanker and RankNet are of the same order. Thus, we do not boost the training time of the model but only the performance. On the other hand the training time beats all the other models. \footnote{For our implementation of the model and the tests see https://github.com/kramerlab/direct-ranker.}

\subsection{Sensitivity on Dataset Properties}\label{sec:results_datasets}
\begin{figure*}[htb]
\centering
\begin{subfigure}[b]{0.4\textwidth}
	\begin{tikzpicture}
	    \draw (0, 0) node[inner sep=0] (pic) {\includegraphics[width=\textwidth]{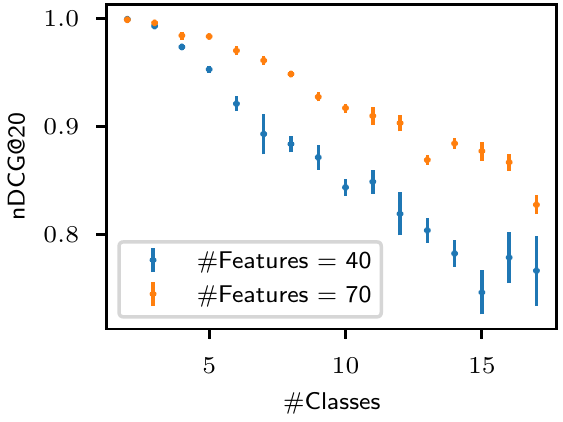}};
	    \draw ($(pic.south west)+(.2,.2)$) node {\parbox[b]{0.5cm}{\caption{}\label{fig:class_dependence}}};
	\end{tikzpicture}
\end{subfigure}
\begin{subfigure}[b]{0.4\textwidth}
	\begin{tikzpicture}
	    \draw (0, 0) node[inner sep=0] (pic) {\includegraphics[width=\textwidth]{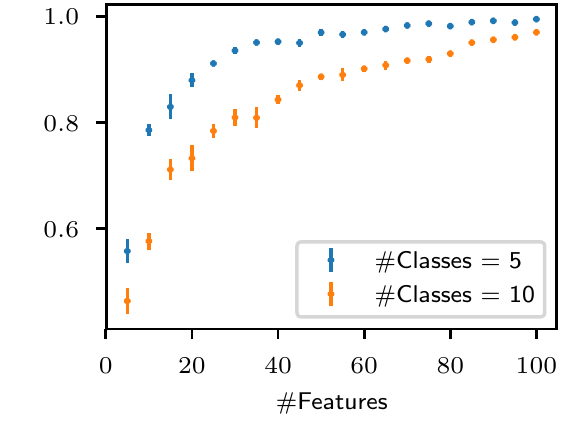}};
	    \draw ($(pic.south west)+(.2,.2)$) node {\parbox[b]{0.5cm}{\caption{}\label{fig:feature_dependence}}};
	\end{tikzpicture}
\end{subfigure}
\begin{subfigure}[b]{0.4\textwidth}
    \begin{tikzpicture}
	    \draw (0, 0) node[inner sep=0] (pic) {\includegraphics[width=\textwidth]{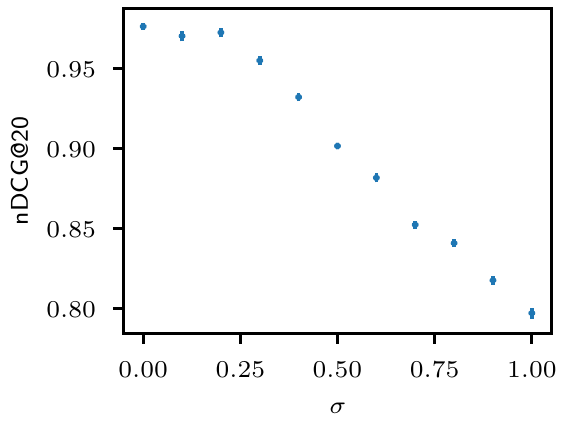}};
	    \draw ($(pic.south west)+(.2,.2)$) node {\parbox[b]{0.5cm}{\caption{}\label{fig:mislabeling}}};
	\end{tikzpicture}
\end{subfigure}
\begin{subfigure}[b]{0.4\textwidth}
	\begin{tikzpicture}
	    \draw (0, 0) node[inner sep=0] (pic) {\includegraphics[width=\textwidth]{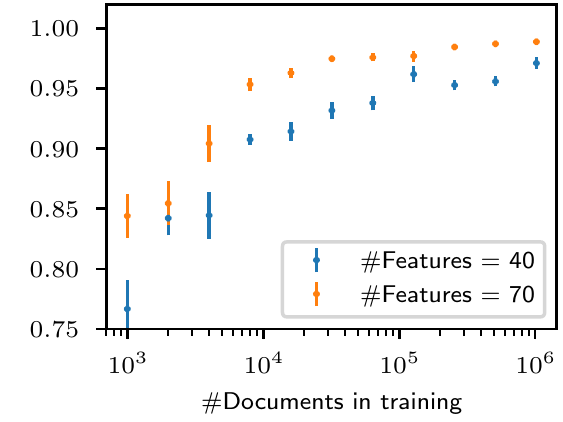}};
	    \draw ($(pic.south west)+(.2,.2)$) node {\parbox[b]{0.5cm}{\caption{}\label{fig:size_dependence}}};
	\end{tikzpicture}
\end{subfigure}
\caption{Plots depicting the sensitivity of the DirectRanker performance on certain data properties, evaluated on the synthetic data (\Cref{sec:setup_datasets}). (\subref{fig:class_dependence}) Dependence on the number of relevance classes ($10^5$ documents in training set). (\subref{fig:feature_dependence}) Dependence on the number of features ($10^5$ documents in training set, five relevance classes). (\subref{fig:mislabeling}) Performance of the DirectRanker with different noise levels on the class labels with 5 classes and 70 features. (\subref{fig:size_dependence}) Dependence on the size of the training set (five relevance classes).}
\end{figure*}
With the following tests we discuss how the DirectRanker performs under different circumstances. The tests were performed as described in \Cref{sec:setup_datasets}. The performance of the DirectRanker was tested for different numbers of relevance classes (\Cref{fig:class_dependence}), features (\Cref{fig:feature_dependence}), for variations of noise on the class labels (\Cref{fig:mislabeling}), and differently sized training sets (\Cref{fig:size_dependence}).
\\ \\
The tests show that, given enough data, our model is able to handle a diverse range of datasets. It especially shows that the DirectRanker can handle many relevance classes as shown in \Cref{fig:class_dependence}. As one would expect, the performance decreases with the number of relevance classes. However, this effect can be counteracted by increasing the size of the training set (see \Cref{fig:size_dependence}) or the number of features (\Cref{fig:feature_dependence}). Additionally, \Cref{fig:mislabeling} shows the robustness of the DirectRanker against noise on the relevance classes. Up to some small noise (approximately $5\%$ mislabeling, i.e. $\sigma = 0.25$), the performance decreases only marginally, but drops significantly for larger noise. Still, even with $50\%$ of the documents being mislabeled (i.e. $\sigma = 0.75$) the NDCG@20 does not drop below 0.80. This suggests that the theoretical findings in \Cref{sec:theory} for ideal data stay valid for real-world data. 

\subsection{Comparing Neighbors in a Sorted List}\label{sec:rank-classification}
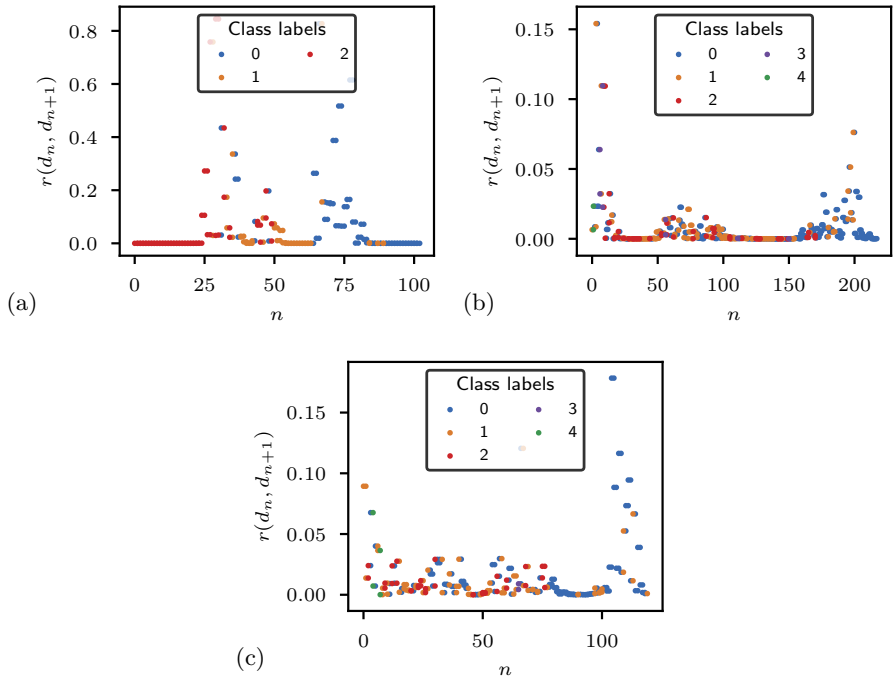
\begin{figure*}[htb]
\centering
\begin{subfigure}[b]{0.4\textwidth}
	\begin{tikzpicture}
	    \draw (0, 0) node[inner sep=0] (pic) {\includegraphics[width=\textwidth]{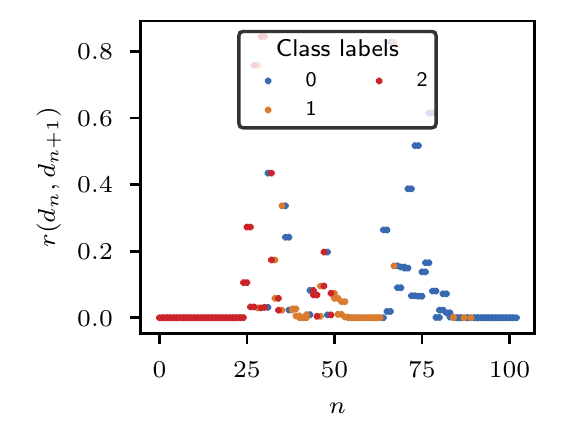}};
	    \draw ($(pic.south west)+(.2,.2)$) node {\parbox[b]{0.5cm}{\caption{}\label{fig:own_sorted_list}}};
	\end{tikzpicture}
\end{subfigure}
\begin{subfigure}[b]{0.4\textwidth}
	\begin{tikzpicture}
	    \draw (0, 0) node[inner sep=0] (pic) {\includegraphics[width=\textwidth]{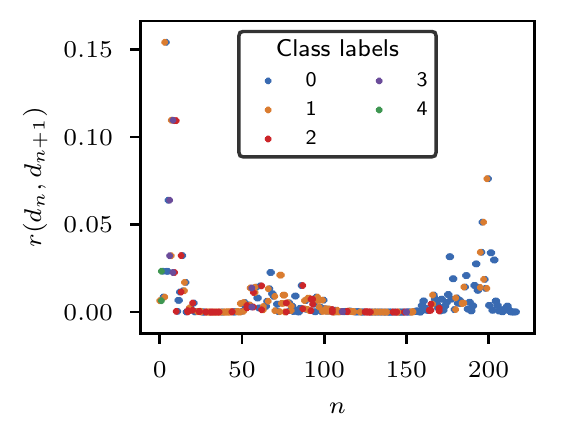}};
	    \draw ($(pic.south west)+(.2,.2)$) node {\parbox[b]{0.5cm}{\caption{}\label{fig:letor_sorted_multi}}};
	\end{tikzpicture}
\end{subfigure}
\begin{subfigure}[b]{0.4\textwidth}
	\begin{tikzpicture}
	    \draw (0, 0) node[inner sep=0] (pic) {\includegraphics[width=\textwidth]{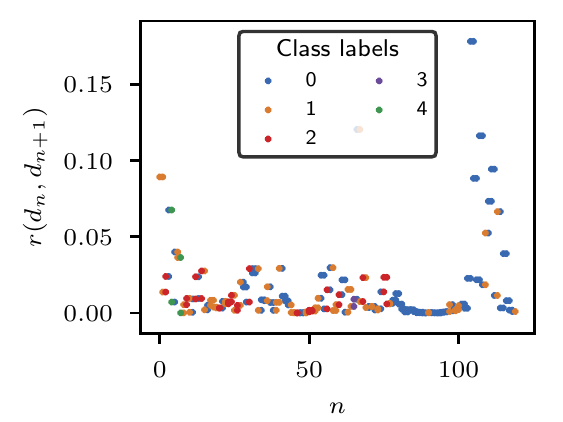}};
	    \draw ($(pic.south west)+(.2,.2)$) node {\parbox[b]{0.5cm}{\caption{}\label{fig:letor_sorted_binary}}};
	\end{tikzpicture}
\end{subfigure}
\caption{Output $r$ of the DirectRanker of two successive documents $d_n, d_{n+1}$ inside a list sorted by the DirectRanker. (\subref{fig:own_sorted_list}) For synthetic data. (\subref{fig:letor_sorted_multi}) For a single query from the MSLR--WEB10K data, unbinarized, five relevance classes. (\subref{fig:letor_sorted_binary}) Same as (\subref{fig:letor_sorted_multi}) (but different query), binarized relevance classes.}
\end{figure*}
Additionally to standard ranking, we suggest that the Direct\-Ranker can be used for classification as well. \Cref{fig:own_sorted_list} depicts the output of the DirectRanker of successive instance pairs inside a sorted list of synthetically generated data. Evidently, the output peaks at the locations on which the relevance classes change. Classes can therefore be separated  by finding these peaks, even on unlabeled data. The same effect can be observed for the MSLR--WEB10K dataset. In \Cref{fig:letor_sorted_multi} the ranker was trained using all five relevance classes, in \Cref{fig:letor_sorted_binary} using the binarized ones. In both of these cases we can again observe the peaks in the ranker output separating the classes, even though the effect is less pronounced compared to the synthetic data. It is remarkable that even in the binarized case the peaks separate the relevance classes relatively well.
\\ \\
It might be feasible to train a classifier by using it in the feature part of the DirectRanker and comparing pairs of instances with each other during training. The trained classifier would be obtained by connecting $o_1$ directly to $nn_1$ keeping the weights. The discussed peaks can be used to find a threshold in $o_1$ separating the classes independent of knowledge about the true label. 
\\ \\
This approach could be an advantage in datasets with little occurrence of one class because the large number of ways to combine two instances of different classes boosts the statistics.

\section{Discussion and Conclusions}\label{sec:conclusion}
The scheme for network structures proposed and analyzed in this paper is a generalization of RankNet: We show which properties of components of RankNet are essential to bring about its favorable behavior and doing so, pave the way for performance improvements. As it turns out, only a few assumptions about the network structures are necessary to be able to learn an order of instances. The requirements on the data for training are also minimal: The method can be applied to discrete and continuous data, and can be employed for simplified training schedules with the comparison of neighboring classes (or other relevant pairs of relevance classes) only. Theoretical results shed some light on the reasons why this is the case. Experiments confirm this and show that the scheme delivers excellent performance also on real-world data, where we may assume that instances are mislabeled with a certain probability. In many recent comparisons, RankNet is shown to exhibit inferior performance, leading to the conclusion that listwise approaches are to be preferred over pairwise approaches. Looking at the experimental results on the LETOR dataset in this paper, there may be reason to reconsider that view. However, it might be interesting to adapt the ideas LambdaRank and LambdaMART for listwise optimization to the DirectRanker model.
\\ \\
Also, it is remarkable that such a simple, transparent approach can match or outperform the performance of more recent and much more complex models, like ES-Rank. Experiments with synthetic data show how the performance can degrade when given more relevance classes, fewer features or fewer training instances. However, these results also indicate how the loss of performance can be compensated by the other factors. Additionally to standard ranking, we suggest that the Direct\-Ranker can be used for classification as well. First tests showed promising results. A more systematic investigation of this is the subject of future work.

\newpage

\subsubsection{Acknowledgement}
We would like to thank Dr. Christian Schmitt for his contributions to the work presented in this paper. 
\\ \\
We also thank Luiz Frederic Wagner for proof(read)ing the mathematical aspects of our model.	
\\ \\
Parts of this research were conducted using the supercomputer Mogon and/or advisory services offered by Johannes Gutenberg University Mainz (hpc.uni-mainz.de), which is a member of the AHRP (Alliance for High Performance Computing in Rhineland Palatinate,  www.ahrp.info) and the Gauss Alliance e.V.
\\ \\
The authors gratefully acknowledge the computing time granted on the supercomputer Mogon at Johannes Gutenberg University Mainz (hpc.uni-mainz.de).
\\ \\
This research was partially funded by the Carl Zeiss Foundation Project: 'Competence Centre for High-Performance-Computing in the Natural Sciences' at the University of Mainz.

\bibliographystyle{splncs04}

\end{document}